\documentclass{interact}

\usepackage{tikz}
\usepackage{tikz-qtree}
\usetikzlibrary{trees} 

\usepackage[noadjust]{cite}
\usepackage{xcolor}
\RequirePackage[all]{xy}

\newtheorem{thm}{Theorem}[section]

\newtheorem{prop}[thm]{Proposition}
\theoremstyle{definition}
\newtheorem{defn}[thm]{Definition}
\theoremstyle{remark}

\numberwithin{equation}{section}

\usepackage{hyperref}

\newcommand{\BibTeX}{B\kern-0.1emi\kern-0.017emb\kern-0.15em\TeX}
\newcommand{\XYpic}{$\mathrm{X\kern-0.3em\raisebox{-0.18em}{Y}}$-$\mathrm{pic}\,$}

\newcommand{\cl}{C \kern -0.1em \ell}  

\newcommand{\hypergeom}[5]{\mbox{$
_{#1} F_{#2}
\!\!\;
\Bigl(
\!\!\!\!
\begin{array}{c}
\multicolumn{1}{c}{\begin{array}{c}
#3
\end{array}}\\[1mm]
\multicolumn{1}{c}{\begin{array}{c}
#4
            \end{array}}\end{array}
\!\!\!\!
; \displaystyle{#5}\Bigr)
$}}
\newcommand{\hypergeomq}[6]{\mbox{$
_{#1} \phi _{#2}
\!\!\;
\Bigl(
\!\!\!\!
\begin{array}{c}
\multicolumn{1}{c}{\begin{array}{c}
#3
\end{array}}\\[1mm]
\multicolumn{1}{c}{\begin{array}{c}
#4
            \end{array}}\end{array}
\!\!\!\!
; \displaystyle{#5}
;\displaystyle{#6}\Bigr)
$}}

\begin{document}
\title{Classification of 2-Orthogonal Polynomials with Brenke Type Generating Functions}
\author{
\name{H. Chaggara\textsuperscript{a}\thanks{E-mail address: hshaggara@kku.edu.sa / hamza.chaggara@ipeim.rnu.tn} and A. Gahami\textsuperscript{b}\thanks{E-mail address: aelgahami@yahoo.fr} }
\affil{\textsuperscript{a}Mathematics
Department, College of Science, King Khalid University, Abha, Saudi Arabia.\\
\textsuperscript{b}Higher School of Sciences and Technology, Sousse University, Tunisia.}
}

\date{\today}

\maketitle
\begin{abstract}
The Brenke type generating functions are the polynomial generating
functions of the form
$$\sum_{n=0}^{\infty}{P_n(x )\over n!}t^n=A(t)B(xt),
$$
where $A$ and $B$ are two formal power series subject to the conditions $A(0)\;B^{(k)}(0)\neq0,\, k=0,1,2\ldots$.\\
In this work, we determine all Brenke-type polynomials when they are also $2$-orthogonal polynomial sets, that is to say, polynomials satisfying one standard four-term recurrence relation. That allows us, on one hand, to obtain new 2-orthogonal sequences generalizing known orthogonal families of polynomials, and on the other hand, to recover particular cases of polynomial sequences discovered in the context of $d$-orthogonality.\\
The classification is based on the resolution of a three-order difference equation (Relation (\ref{reccur 2 orthog (2)})) induced by the four-term recurrence relation satisfied by the considered polynomials.
This study is motivated by the
work of Chihara~\cite{chihara1968orthogonal} who gave all pairs $(A (t), B(t))$ for which $\{P_n(x)\}_{n\geq 0}$ is an orthogonal polynomial sequence.\\

\textit{\textbf{2010 MATHEMATICS SUBJECT CLASSIFICATIONS:}}
33C45; 39A70; 41A10; 41A58.
\end{abstract}
\section{Introduction}
Let $\mathcal{P}$ be the linear space of polynomials with complex coefficients and let $\mathcal{P'}$ be its algebraic dual. 
A polynomial sequence $\{P_n\}_{n\geq0}$ in $\mathcal{P}$ is called a {\it polynomial set} (PS), if $\deg P_n(x)=n$ for all $n$.
By orthogonality of polynomials sequences, we mean orthogonality with respect to a linear functional in $\mathcal{P'}$~\cite{chihara2011introduction}.\\
{\it Multiple orthogonal polynomials} are a generalization of {\it standard orthogonal polynomials}. Their orthogonality
measures are spread across a vector of $d$ measures and they are polynomials on a single variable
depending on a multi-index $\textbf{n}=(n_0,n_1,\ldots,n_{r-1})$. There are two types of multiple orthogonal polynomials: Type I and Type II.
This notion has many applications in various fields of mathematics  as approximation theory, analytic number theory, spectral theory of operators, special functions theory, and vector and simultaneous Pad\'{e} approximants, approximants~\cite{van1987vector,de1985simultaneous,
van1987approximants}.
A Riemann-Hilbert approach to such polynomials appears in \cite{vanassche2001}. For further information about multiple orthogonal polynomials, we refer to~\cite{aptekarev1998multiple} and Chapter 23 written by Van Assche in Ismail's book~\cite{ismail2005classical}.

Our contribution deals with the notion of {\it $d$-orthogonality}, considered as an interesting subclass of type II multiple orthogonal polynomials whose multi-index lie on the step-line near the diagonal. This concept was introduced by Van Iseghem~\cite{van1987vector} and Maroni~\cite{maroni1989orthogonalite}, where linear recurrence relations of higher order are studied for sequences of monic polynomials in such a way a generalization of the standard orthogonality is given when $d=1$.
$d$-Orthogonality has received increasing interest from the pioneering work by P. Maroni~\cite{maroni1989orthogonalite}, where the algebraic properties of a sequence of orthogonal polynomials with respect to a $d$-dimensional vector of linear functionals are studied. Intensive work has been done from such a time in the analysis of the $d$-orthogonality counterpart of classical orthogonal polynomials. In particular, we must point out \cite{
cheikh2000classical,
cheikh2008some,
NBcheikh2014d,
boukhemis05,
BM,
boukhemis97,
douakcaracter,
Douakmaroni,
douak1996,
lamiri2013d,
cheikh2007dunkl,
cheikh2011characterization,
BenCheikhGam2}, as well as~\cite{
douak1992,
douakmaroni20,
douak99,
bencheikh00,
marcellan2021,
chaggara2018a}, where examples of classical $2$-orthogonal polynomials are studied from the perspective of the characterizations of the standard classical ones. New polynomial sets which contain classical orthogonal polynomials have been created so far.
Furthermore, the notion of $d$-orthogonality is related to the study of vectorial continued fractions, nonlinear automorphisms of the Weyl algebra~\cite{vinet2009automorphisms}, Darboux transformations in~\cite{marcellan2021}, birth and death process~\cite{zerouki2006}.

In this work, we will be interested in a particular class of polynomials, but nevertheless important. These polynomials are defined by their {\it Brenke type generating function (GF)} as follows \cite{chihara1968orthogonal}
\begin{equation}\label{form brenke}
A(t)B(xt)=\sum_{n=0}^{\infty}\frac{P_n(x)}{n!}t^n,
\end{equation}
where $A$ and $B$ are two formal power series satisfying:
\begin{equation}\label{expression A et B}
A(t)=\sum_{k=0}^{\infty}a_kt^k,\quad B(t)=\sum_{k=0}^{\infty}b_kt^k,\quad a_0b_k\neq0, \ k=0,1,\ldots.
\end{equation}
Thus the corresponding {\it Brenke polynomials} are
\begin{equation}\label{explicit Brenke}
P_n(x)=\sum_{k=0}^na_{n-k}b_kx^k,\ n=0,1,\ldots.
\end{equation}
In~\cite{chihara1968orthogonal}, Chihara made a complete classification of orthogonal polynomials with Brenke-type GF. The resulting classes consisted of seven families (known as Brenke-Chihara polynomials), among them, Laguerre, Szeg\"{o}-Chihara generalized Hermite, Wall, Al-Salam-Carlitz, Stieltjes-Wigert polynomials as well as a "new" set of orthogonal polynomials (Chihara PS as suggested in~\cite{askey2001}). The corresponding orthogonality measure, which has a mass on the whole real line and yet is not symmetric, was given explicitly by Chihara in~\cite{chihara1971orthogonality}. \\
The problem of determining orthogonal PS with Brenke type GF is a particular case of Geronimus problem who considered OPS of the form 
$$P_n(x)=a_nb_0+\sum_{k=1}^n a_{n-k}b_k\prod_{i=1}^k (x-x_i).$$
Geronimus obtained a set of necessary and sufficient conditions on $\{a_n\}_{n\geq 0},
 \{b_n\}_{n\geq 0}$ and $\{x_n\}_{n\geq 1}$ and exhibited a number of special cases explicitly. However, this problem has remained unsolved in its full generality. The Chihara problem consists of the particular case $x_k=0,\ k=1,2,\ldots$.\\
Brenke polynomials are a particular case of a larger class of polynomials known as \textit{generalized Appell polynomial sequences} or \textit{Boas-Buck polynomial sequences}~\cite{boas1958,ismail1974obtaining}. Brenke polynomials are also reduced to {\it Appell} and {\it $q$-Appell} ones when $B=\exp$ and $B=e_q$, respectively.
Brenke polynomials generated many well-known polynomials in the literature, namely 
monomials, Laguerre, Generalized Hermite,
Generalized Gould-Hopper, Appell-Dunkl, $d$-Hermite, $d$-Laguerre, Bernoulli, Euler, Al-Salam-Carlitz, $q$-Laguerre, discrete $q$-Hermite PS,\ldots. \\
Brenke polynomials appear in many areas of mathematics such as approximation theory, summability, and differential equations~\cite{ansari2019approximation}. 
Furthermore, they play a central role in the determination of all MRM-triples associated with Brenke-type GF~\cite{asai2013brenke}. 
Sz\'{a}sz operators are extensions of Bernstein operators to infinite intervals. They play a very important role in the field of approximation theory~\cite{varma12}. Brenke polynomials provide a crucial tool to extend Sz\'{a}sz operators~\cite{varma12,braha2020some}. 
For instance, a Chlodowsky type generalization of Sz\'{a}sz operators defined by means of Brenke type polynomials are investigated in~\cite{mursaleen15}. 
Moreover, an approximation by Durrmeyer type Jakimovski–Leviatan operators involving Brenke type polynomials is established in~\cite{wani21}. 

The notion of orthogonal polynomials relative to a regular linear functional has been broadened as a natural extension to the so-called $d$-orthogonal polynomials ($d$ being a non-negative integer) by Van Iseghem~\cite{van1987vector} and Maroni~\cite{maroni1989orthogonalite}. These are polynomials in one variable satisfying a specific orthogonality relations
with respect to $d$ different linear functionals. \\
Any $\{P_n\}_{n\geq 0}$ is said to be $d$-orthogonal polynomial set ($d-$OPS) if there exists a $d$-dimensional vector of linear functionals,
$\mathbb{U}=(u_0,\ldots,u_{d-1})$ such that:
\begin{equation}\label{eqdops}
\begin{cases} \langle u_k, x^m P_n\rangle =0\quad \textrm{if}\;\; n> md+k,\ m=0,1,\ldots,\\
\langle u_k, x^m P_{md+k}\rangle \neq 0 \quad\; m=0,1,\ldots;\;k=0,1,\ldots,d-1,
\end{cases}
\end{equation}
 where $\langle u,f\rangle$ denotes the effect of a linear functional $u\in \mathcal{P}^\prime $ on the polynomial $f\in\mathcal{P}$.\\ 
A relevant characterization of $d$-OPS by a $(d+1)$-order recurrence relation was given by Van Iseghem~\cite{van1987vector} and Maroni~\cite{maroni1989orthogonalite}.
\begin{equation}\label{recuurence orthogonality}
xP_n(x)=\sum_{k=-1}^d\gamma_k(n)P_{n-k}(x),\quad \gamma_k(n)\in\mathbb{C},\ k,n=0,1,2,\ldots,
\end{equation}
with regularity condition  $\gamma_d(n)\cdot\gamma_{-1}(n)\neq0$, and the convention $P_{-k}=0,\; k=1,2,\ldots$.\\
For the particular case $d=1$, this result is reduced to the well-known Shohat-Favard Theorem (or Spectral Theorem for OPS)~\cite{chihara2011introduction}.\\
Most of the known $d$-orthogonal families were introduced either as already known families of polynomials satisfying higher order recurrence relation (Gould-Hopper, Humbert polynomials,\ldots) or as solutions of characterization problems. Such problems consist in finding all $d$-OPS having a given property namely, $d$-OPS of Appell type~\cite{cheikh2007dunkl, cheikh2011characterization,douak1996}, $q$-Appell type~\cite{zaghouani2005some},
Specific $d$-OPS of Sheffer type~\cite{bencheikh2007},
$2$-OPS of Sheffer type~\cite{BM,chaggara2018a}, $d$-OPS of Sheffer type~\cite{BenCheikh-Gam, chaggara2018a, varma2018} and $(d+1)$-fold symmetric $d$-OPS of Sheffer type~\cite{chaggara2018a}, classical and semi-classical character~\cite{BM,saib2013semi}, specific hypergeometric form~\cite{cheikh2008some,lamiri2013d},\ldots\\
The characterization of $d$-OPS with Brenke-type GF was considered by many authors. The special cases, $B =~ \exp,\ A = \exp,\ B=e_q,$ and $B=\exp_\mu$: the generalized exponential function, were investigated, respectively, by Douak \cite{douak1996}, Ben Cheikh and Douak~\cite{bencheikh2000,cheikh2000classical}, Zaghouani~\cite{zaghouani2005some} and Ben Cheikh and Gaied~\cite{cheikh2007dunkl}. Furthermore, the $(m+1)$-fold symmetry property was studied in~\cite{NBcheikh2014d}.
Recently, we characterized all $d$-orthogonal polynomial sets of Brenke type~\cite{gahami23}. We obtained several new and known results. We give some examples as illustrations. Our main result was the following characterization theorem 
\begin{thm}\label{th}\cite[Theorem]{gahami23}
Let $\{P_n\}_{n\geq0}$ be a PS of Brenke type generated by (\ref{form brenke}) then, $\{P_n\}_{n\geq 0}$ is a $d$-OPS if and only if
\begin{equation}\label{recurrence relation}
    \sum_{i=0}^m\Bigl(\sum_{j=i}^ma_{j+1}\widehat{a}_{m-j}\Bigr)\Delta_{n-i}=0,\quad m=n,n+1,\ldots, d+1,
\end{equation}
with the regularity condition
\begin{equation}\label{regularity condition}
   \sum_{i=0}^d\Bigl(\sum_{j=i}^da_{j+1}\widehat{a}_{d-j}\Bigr)\Delta_{n-i}\neq0,
\end{equation}
where $\widehat{A}(t)=\frac{1}{A(t)}=\sum_{k=0}^\infty\widehat{a}_kt^k$ and $\Delta_n=r_n-r_{n-1},\ r_n=\frac{b_n}{b_{n+1}}$ and $r_{-1}=0$.
\end{thm}
Our purpose in this work is to give explicitly all 2-OPS of Brenke type.
Such a classification takes into account the fact that PS which is obtainable
from one another by an affine change of the variable or by an alternative normalization are considered to be
equivalent.\\
Our main tool is the recurrence relation (\ref{recurrence relation}), the GF (\ref{form brenke}), and the resolution of the difference equation~(\ref{reccur 2 orthog (2)}.
Some polynomial sequences among the obtained families represent a generalization of these obtained by Chihara and either are particular cases of some known $d$-OPS of Brenke type or appear to be as new PS.\\
Before discussing the above problem, let us recall some definitions which we need below.\\
The Pochhammer symbol $\{(x)_n\}_{n\geq0}$ and the $q$-shifted factorial  $\{(x;q)_n\}_{n\geq0}$ are defined as follows:
 $$(x)_n= \prod_{k=0}^{n-1}(x+k) 
              \ \textrm{and}\ (x;q)_n=
               \prod_{k=0}^{n-1}(1-xq^k),\ n=1,2,\ldots,\ (x)_0=(x;q)_0=1.$$
The standard notion for hypergeometric function is
\begin{equation*}
\hypergeom{p}{q}{(a_p)}{(b_q)}{t}=\sum_{n=0}^\infty\frac{(a_1)_n\ldots(a_p)_n}{(b_1)_n\ldots(b_q)_n}\frac{t^n}{n!},
\end{equation*}
where the contracted notation $(a_p)$ is used to abbreviate the array of $p$ parameters $a_1,\ldots , a_p$.\\
The $q$-hypergeometric series are defined by
\begin{equation*}
\hypergeomq{r}{s}{(a_r)}{(b_s)}{q}{t}
=\sum_{n=0}^\infty\frac{(a_1;q)_n\cdots (a_r;q)_n}{(b_1;q)_n\cdots(b_s;q)_n}
\frac{((-1)^nq^{n\choose2})^{1+s-r}}{(q;q)_n}t^n,
\end{equation*}
An elementary $q$-hypergeometric series is the $q$-exponential series defined by
\begin{equation*}
 e_q(t)=\frac{1}{(t;q)_\infty}=\,_1\phi_0(_{-}^{0};q,t),\ (t;q)_\infty=\prod_{k=0}^{\infty}(1-tq^k),\ 0<q<1.
\end{equation*}
The following relation holds
\begin{equation*}
(-t;q)_\infty=\frac{1}{e_q(-t)}=\,_0\phi_0(_{-}^{-};q,-t).
\end{equation*}

\section{Some Results About Characterization of Brenke Type $d$-OPS}
In this section, we collect the essential results, other than the Theorem stated in Section 1, that allow us to classify all 2-OPS with Brenke-type GF. 
Without loss of the generality, we can suppose in (\ref{expression A et B}) that $a_0=b_0=1$.\\ 
A straightforward and interesting consequence of the above theorem is the following statement:
\begin{prop}\label{cor d-orth}\cite[Corollary]{gahami23}
A necessary condition for the $d$-orthogonality of $\{P_n\}_{n\geq0}$ is obtained
 when we take $m=d+1$ in (\ref{recurrence relation})
\begin{equation}\label{d+1 ordre reccu}
   \Bigl(\sum_{j=0}^{d+1}a_{j+1}\widehat{a}_{d+1-j}\Bigr)\Delta_n+ \Bigl(\sum_{j=1}^{d+1}a_{j+1}\widehat{a}_{d+1-j}\Bigr)\Delta_{n-1}+\cdots+a_{d+2}
   \Delta_{n-d-1}=0.
\end{equation}
\end{prop}
Using (\ref{recuurence orthogonality}) and the above theorem, we have the following 
\begin{prop}\label{proposition regularity}
If $\{P_n\}_{n\geq 0}$ is a $d$-OPS, then $(d+1)$ consecutive numbers $\Delta_n$ cannot be zero and $(d+1)$ consecutive coefficients $a_n$ cannot be zero.
\end{prop}
\begin{proof}
From (\ref{regularity condition}) we deduce immediately that (d+1) consecutive numbers of $\Delta_n$ cannot be zero.\\
On the other hand, equating the coefficients of $x^m$ in both sides in (\ref{recuurence orthogonality}), we obtain
\begin{equation*}
\gamma_d(n)(n-d)!a_{n-m-d}=n!a_{n-m+1}(\frac{b_{m-1}}{b_m}-(n+1)\gamma_{-1}(n))-\sum_{k=0}^{d-1}\gamma_k(n)(n-k)!a_{n-m-k}.
\end{equation*}  
Since $\gamma_d(k)\neq0$, we deduce that $(d+1)$ consecutive numbers of $a_n$ cannot be zero.
\end{proof}
An extension of the notion of the symmetric polynomials is the so-called $(m+1)-$fold symmetric polynomials defined as follows: 
\begin{defn}
A polynomial sequence $\{P_n\}_{n\geq0}$ is called $(m+1)$-fold symmetric polynomial set if it fulfills 
\begin{equation}\label{def d_symetric}
    P_n(w_mx)=w_m^nP_n(x),\quad w_m=\exp(\frac{2i\pi}{m+1}),\ n=0,1,2\ldots.
\end{equation}
\end{defn}
We denote that if $\{P_n\}_{n\geq 0}$ is a $PS$ of Brenke type generated by (\ref{form brenke}), $\{P_n\}_{n\geq 0}$ is $m+1$-fold symmetric if and only if $A(t)=A(w_mt)$ and consequently the GF (\ref{form brenke}) can be written as 
\begin{equation}\label{generating func d-symm}
G(x,t)=A_1(t^{m+1})B(xt),
\end{equation}
where $A_1(t)=\sum_{n=0}^\infty a_{1,n}t^n$ and $a_{1,0}=1$. 
\begin{prop}\label{cor recuu d-symmet}\cite[Corollary 4.3]{gahami23}
Let $\{P_n\}_{n\geq 0}$ be a $d$-OPS of Brenke type generated by (\ref{form brenke}). Then
$\{P_n\}_{n\geq 0}$ is $(m+1)$-fold symmetric if and only if $a_1=\cdots =a_m=0$
\end{prop}
In ~\cite{NBcheikh2014d}, Ben Cheikh and Ben Romdhane investigated $(d+1)$-fold symmetric, $d-$OPS of Brenke type. They showed that there exist exactly two classes of families of this kind.
\begin{prop}\cite[Theorem 2.1]{NBcheikh2014d}\label{corllary carac d-orth d-symm}
Let $\{P_n\}_{n\geq0}$ be a $(m+1)$-fold symmetric polynomials set generated by (\ref{generating func d-symm})
then $\{P_n\}_{n\geq0}$ is $d-$orthogonal if and only if $A_1(t)$ and $B(t)$ are one of the two possibilities 
  $$\bullet\ G_1(x,t)=\sum_{u=0}^d\frac{\exp(a_{1,1}t)(xt)^u}{\prod_{i=0}^{u-1}r_i}
  \hypergeom{0}{d}{-}{{r_0\over v_0}+1,\ldots,{r_{u-1}\over v_{u-1}}+1,{r_u\over v_u},\ldots,{r_{d-1}\over v_{d-1}}}{\frac{(xt)^{d+1}}{v_0\cdots v_d}},$$
  where $v_i\neq0,\ i=0,\ldots,d;\ \frac{r_i}{v_i}\neq -1,-2,\ldots;\ i=0,\ldots,(d-1),\ \prod_{i=0}^{-1}=1$,
and 
  $$\bullet\ G_2(x,t)=\sum_{u=0}^d\frac{e_q(a_{1,1}(1-q)t)(xt)^u}{\prod_{i=0}^{u-1}s_i(t_i-1)}
 \hypergeomq{0}{d}{-}{{qt_0,\ldots,qt_{u-1},t_u,\ldots t_{d-1}}}{q}{\frac{q^{u+1}(xt)^{d+1}}{s_0\cdots s_d}}, $$
where $s_i\neq0;\ i=0\ldots,d$, and $t_i\neq q^{-k},\ k\in\mathbb{N},\,i=0,\ldots, (d-1)$.
\end{prop}
\section{Classification of 2-orthogonal polynomials set of Brenke type}
\subsection{Necessary conditions}
Let $\{P_n\}_{n\geq 0}$ be a PS generated by (\ref{corllary carac d-orth d-symm}). Taking $d=2$ in the Proposition~\ref{cor d-orth},\\ we obtain the following necessary condition ensuring the 2-orthogonality of $\{P_n\}_{n\geq 0}$.
\begin{equation}\label{reccur orthog d=2 (1)}
    (a_1\widehat{a}_3+a_2\widehat{a}_2+a_3\widehat{a}_1+a_4)\Delta_n+(a_2\widehat{a}_2+a_3\widehat{a}_1+a_4)\Delta_{n-1}+(a_3\widehat{a}_1+a_4)\Delta_{n-2}+a_4\Delta_{n-3}=0,
\end{equation}
with the regularity condition 
\begin{equation}\label{regularity condition1}
(a_1\widehat{a}_2+a_2\widehat{a}_1+a_3)\Delta_n+(a_2\widehat{a}_1+a_3)\Delta_{n-1}+a_3\Delta_{n-2}\neq0,\ n=3,4,\ldots.
\end{equation}
Taking into account that $\sum_{k=0}^n\widehat{a}_ka_{n-k}=0,\ n=1,2,\ldots$, relations (\ref{reccur orthog d=2 (1)}) and (\ref{regularity condition1}) become respectively:
\begin{eqnarray}\label{reccur 2 orthog (2)}
    (-a_1^4-2a_1a_3+3a_1^2a_2-a_2^2+a_4)\Delta_n+(a_1^2a_2-a_2^2-a_1a_3+a_4)\Delta_{n-1}\nonumber\\
    +(a_4-a_1a_3)\Delta_{n-2}+a_4\Delta_{n-3}=0,\  n=3,4,\ldots\\
\label{regularity condition 2}
    (a_1^3-2a_1a_2+a_3)\Delta_n+(a_3-a_1a_2)\Delta_{n-1}+a_3\Delta_{n-2}\neq0,\ n=3,4,\ldots
\end{eqnarray}
\subsection{3-fold symmetric polynomials sequence}
According to Proposition~\ref{cor recuu d-symmet}, $\{P_n\}_{n\geq 0}$ is a 3-fold symmetric PS if and only if \\ $(a_1,a_2)=(0,0)$.
By virtue of Proposition~\ref{corllary carac d-orth d-symm}, we get the  following two GF
\begin{eqnarray}
G_1(x,t)&=&e^{at^3}\Bigl[\hypergeom{0}{2}{-}{{r_0\over v_0}+1,{r_{1}\over v_{1}}+1}{\frac{(xt)^3}{v_0v_1v_2}}+{xt\over r_0}\hypergeom{0}{2}{-}{{r_0\over v_0}+1,{r_{1}\over v_{1}}}{\frac{(xt)^3}{v_0v_1v_2}}\nonumber\\
&+&{(xt)^2\over r_0r_1} \hypergeom{0}{2}{-}{{r_0\over v_0},{r_{1}\over v_{1}}}{\frac{(xt)^3}{v_0v_1v_2}}\Bigr]\label{GF2Symm1}\\
G_2(x,t)&=&e_q(a(1-q)t^3)\Bigl[
\hypergeomq{0}{2}{-}{qt_0,qt_{1}}{q}{\frac{q(xt)^3}{s_0s_1s_2}}
+{xt\over s_0(t_0-1)}\hypergeomq{0}{2}{-}{qt_0,t_{1}}{q}{\frac{q^{2}(xt)^3}{s_0s_1s_2}}\nonumber\label{GF2Symm2}\\
&+&{(xt)^2\over s_0s_1(t_0-1)(t_1-1)}\hypergeomq{0}{2}{-}{t_2,t_1}{q}{\frac{q^{3}(xt)^3}{s_0s_1s_2}}\Bigr],
\end{eqnarray}
where $t_j\neq1,q^{-1},q^{-2},\ldots,\ \frac{r_j}{v_j}\neq-1,-2,\ldots,$ and $v_j,s_j\neq 0$ for  $j=0,1,2$. \\
The aforementioned families contain many known sequences of 3-fold symmetric 2-OPS, namely the 2-OPS of Generalized Hermite type PS~\cite{cheikh2007dunkl}, the 2-OPS of $q$-Hermite type~\cite{bencheikh14,rajkovic2001q} and the 2-OPS of discrete $q$-Hermite I $\&$ II type~\cite{lamiri13}.
\subsection{The non symmetric case}
According to the above section, a polynomials sequence $\{P_n\}_{n\geq 0}$ with GF~(\ref{form brenke}) is not 3-fold symmetric if and only if $(a_1,a_2)\neq(0,0)$. The key element to studying these polynomials is to investigate the characteristic equation associated with (\ref{reccur 2 orthog (2)}).\\
Suppose that the leading coefficient of (\ref{reccur 2 orthog (2)}) is not 0 and  let $x_1,x_2,x_3$ be the roots of the following equation 
\begin{equation}\label{equation charc}
    (-a_1^4-2a_1a_3+3a_1^2a_2-a_2^2+a_4)r^3+(a_1^2a_2-a_2^2-a_1a_3+a_4)r^2+(a_4-a_1a_3)r+a_4=0.
\end{equation}
Then the solutions of (\ref{reccur 2 orthog (2)}) take one of the following forms  
\begin{equation}\label{form delta generale}
\left\{
\begin{array}{llll}
\Delta_n =& (\alpha+\beta n+\gamma n^2)x_1^n, \quad &\textrm{if}&\;x_1=x_2=x_3 \\
\Delta_n =&(\alpha+\gamma n)x_1^n+\beta x_2^n,\quad&\textrm{if}&\; x_1\neq x_2\;\textrm{and}\;x_2=x_3\\
\Delta_n =& \alpha x_1^n+\beta x_2^n+\gamma x_3^n, \quad&\textrm{if}&\; x_i\neq x_j,\ i\neq j  .
\end{array}
\right.
\end{equation}
We note that (\ref{form delta generale}) remains valid in the case $-a_1^4-2a_1a_3+3a_1^2a_2-a_2^2+a_4=0$, by taking $\gamma=0$.\\
The two cases $a_1=0$ and $a_1\neq 0$ need to be considered separately.
\subsubsection{Case $\mathcal{A}:\ a_1=0$}
In this case, we have $a_2\neq0$.
The relations (\ref{reccur 2 orthog (2)})-(\ref{regularity condition 2}) become respectively
\begin{eqnarray}\label{rr when a_1=0 (d=2)}
 (a_4-a_2^2)(\Delta_n+\Delta_{n-1})+a_4(\Delta_{n-2}+\Delta_{n-3})=0,\ n =3,4,\ldots,\\
\label{regularity condition case1}
 a_3(\Delta_n+\Delta_{n-1}+\Delta_{n-2})\neq0,\ n=2,3,\ldots.
\end{eqnarray} 
This leads to,
\begin{equation}
\label{forme delta a_1=0}
\Delta_n =\alpha(-1)^n,\ \textrm {or}\
\Delta_n = \alpha+(\beta+\gamma n)(-1)^n,\ \textrm{or}\
\Delta_n = \alpha(-1)^n+\beta q^n+\gamma(-q)^n.
\end{equation}
The relation between coefficients and roots in (\ref{rr when a_1=0 (d=2)}) leads to $q^2=\frac{a_4}{a_2^2-a_4}$\\
Next, we discuss all possible cases in accordance with the forms of $\Delta_n$ given in~(\ref{forme delta a_1=0}).
\paragraph{$(\mathcal{A}1):\ \Delta_n=\alpha(-1)^n,\ \alpha\neq0$}.\\
Substitute $\Delta_n$ into (\ref{recurrence relation}), we obtain
$\frac{A(t)}{A(-t)}=1+2a_3t^3$.\\
Therefore 
$\frac{A(-t)}{A(t)}=\frac{1}{1+2a_3t^3}=1-2a_3t^3.$ That is, $(1+2a_3t^3)(1-2a_3t^3)=1,$
which implies that $a_3=0$. This contradicts (\ref{regularity condition case1}).
\paragraph{($\mathcal{A}2):\  \Delta_n=\alpha+(\beta+\gamma n)(-1)^n$}\
\\
$\ast$ {If $(\beta,\gamma)\neq(0,0)$}, is the same way as case ($\mathcal{A}1$), there is no solution.\\
$\ast$ {If $(\beta,\gamma)=(0,0)$}, then $\alpha\neq0$, $r_n=(n+1)\alpha$ and $b_n=\frac{1}{n!\alpha^n}$.\\
Substitute $\Delta_n$ into (\ref{recurrence relation}), we obtain 
$\frac{A'(t)}{A(t)} = 2a_2t+3a_3t^2$.
We conclude that
\begin{equation}\label{exp A a_1=0 cas1}
G(x,t)=\exp(a_2t^2+a_3t^3)\exp(\frac{xt}{\alpha}).
\end{equation}
Here, $\{P_n\}_{n\geq 0}$ is the 2-OPS of Hermite type discovered by Douak in~\cite{douak1996} as the only $2$-OPS of Appell type.
\paragraph{($\mathcal{A}3):\ \Delta_n=\alpha(-1)^n+\beta q^n+\gamma(-q)^n,\  |q|\neq0,1$}\
\\
$\ast$ {If $\alpha\neq0$}, similarly to ($\mathcal{A}1$), we don't have a solution.\\
$\ast$ {If $(\alpha,\gamma)=(0,0)\,\textrm{and}\,\beta\neq0$ or $(\alpha,\beta)=(0,0)\,\textrm{and}\,\gamma\neq0$}, then $\Delta_n=\beta q^n,\; r_n=\beta \frac{1-q^{n+1}}{1-q}$ and $ b_n= \frac{(1-q)^n}{\beta^n(q;q)_n}$.
Substitute $\Delta_n$ into (\ref{recurrence relation}), we get 
$$ \frac{A(t)}{A(qt)}= 1+a_2(1-q^2)t^2+a_3(1-q^3)t^3=(1-\lambda t)(1-\mu t)(1-\nu t).$$
It follows that
\begin{equation}\label{exp A a_1=0 cas2}
G(x,t)=\frac{1}{e_q(-\lambda t;q) e_q(-\mu t;q)e_q(-\nu t;q)} e_q((1-q){xt\over\beta}),
\end{equation}
This family is exactly the $2$-OPS $q$-Appell PS (Brenke type with $B=e_q$) treated by Zaghouani 
in~\cite{zaghouani2005some}.\\ 
$\ast$ {If $\alpha=0,\ \beta\neq0$ and $\gamma\neq0$}, replace $\Delta_n=\beta q^n+\gamma(-q)^n$ in (\ref{recurrence relation}) we get
\begin{equation*}
\frac{A(t)}{A(qt)}=1+a_2(1-q^2)t^2+a_3(1-q^3)t^3\ \textrm{and}\
 \frac{A(t)}{A(-qt)}=1+a_2(1-q^2)t^2+a_3(1+q^3)t^3.
\end{equation*}
It follows that
\begin{eqnarray*}
 a_n&=&a_nq^n+a_2(1-q^2)q^{n-2}a_{n-2}+a_3(1-q^3)q^{n-3}a_{n-3}, \\
a_n&=&a_n(-q)^n+a_2(1-q^2)(-q)^{n-2}a_{n-2}+a_3(1+q^3)(-q)^{n-3}a_{n-3},\ n=3,4\ldots.
\end{eqnarray*}
This implies that $a_3q^{2n-3}a_{2n-3}=0, \ n=2,3\ldots$. Taking $n=3$, we get $a_3=0$, which contradicts the regularity condition (\ref{regularity condition case1}).
\subsubsection{ ($\mathcal{B}):\ a_1\neq0$}\
Next, we discuss the different expressions of $\Delta_n$ given by (\ref{form delta generale}) in accordance with the different expressions of $\Delta_n$ in (\ref{equation charc}).
\paragraph{$(\mathcal{B}1)$: $\Delta_n=\alpha x_1^n+\beta x_2^n+\gamma x_3^n,\ x_1,x_2,x_3$ are pairwise distinct}\
\subparagraph*{$\mathbf{\bullet\ (\mathcal{B}{11}):\ \alpha\beta\gamma x_1x_2x_3\neq0\ and \ x_i\neq1, i=1,2,3}$.}
Substituting $\Delta_n$ into (\ref{recurrence relation}) and taking into account that
$(x_1^n)_n,\ (x_2^n)_n$ and $(x_3^n)_n$ are linearly independent vectors, we deduce that $\{P_n\}_{n\geq 0}$ is a 2-OPS if and only $A$ satisfies
\begin{equation}\label{expression of Q_1,Q_2,Q_3}
\frac{A(t)}{A(x_it)}=1+\delta_{i,1}t+\delta_{i,2}t^2+\delta_{i,3}t^3 =Q_i(t),\ i=1,2,3,
\end{equation}
where 
\begin{eqnarray}\label{expression delta_ij}
\delta_{i,1}&=&a_1(1-x_i),\,\delta_{i,2}=a_1^2(x_i^2-x_i)+a_2(1-x_i^2)\nonumber\\\delta_{i,3}&=&a_1^3(x_i^2-x_i^3)-a_1a_2(x_i^2+x_i-2x_i^3)+a_3(1-x_i^3),\ i=1,2,3
\end{eqnarray}
The expression (\ref{expression of Q_1,Q_2,Q_3}) provides the following necessity
\begin{equation}\label{relation of Q_1,Q_2,Q_3}
 Q_i(x_jt)Q_j(t)=Q_i(t)Q_j(x_it),\quad 1\leq i\neq j\leq3
\end{equation}
Comparing the coefficients in both sides of (\ref{relation of Q_1,Q_2,Q_3}), we distinguish four cases.
\\
\underline{${\mathbf(\mathcal{B}{111}):\ \delta_{1,3}\delta_{2,3}\delta_{3,3}\neq0}$.} Here we have $x_1^3=x_2^3=x_3^3$, so 
 $(x_1,x_2,x_3)=(q,jq,\bar{j}q),$\\
 $\ q\neq 1,j,\bar{j},$\ where $j=e^{i\frac{2\pi}{3}}$. \\
$\triangleright$ \textit{Determination of $B(t)$ and $A(t)$}. We have
$$
  \Delta_{3n} = (\alpha+\beta+\gamma)q^{3n}, \
  \Delta_{3n+1} = (\alpha+\beta j+\gamma\bar{j})q^{3n+1}, \
  \Delta_{3n+2} = (\alpha+\beta\bar{j}+\gamma j)q^{3n+2},
$$
and
\begin{eqnarray*}
  r_{3n} &=& \frac{\alpha+\beta+\gamma}{1-k_1q}\left(1-k_1q^{3n+1}\right) ,\; 
  r_{3n+1} =\frac{\alpha+\beta\bar{j}+\gamma j}{k_2-q}\left(1-k_2q^{3n+2}\right),\; 
  r_{3n+2} =  k_3\frac{1-q^{3n+3}}{1-q^3},
\end{eqnarray*}
where 
\begin{eqnarray*}
  k_1 &=& \frac{(\alpha+\beta j+\gamma\bar{j})+(\alpha+\beta\bar{j}+\gamma j)q+(\alpha+\beta+\gamma)q^2}
   {(\alpha+\beta+\gamma)+(\alpha+\beta j+\gamma\bar{j})q+(\alpha+\beta\bar{j}+\gamma j)q^2} \\
k_2 &=& \frac{(\alpha+\beta\bar{j}+\gamma j)+(\alpha+\beta+\gamma)q+(\alpha+\beta j+\gamma\bar{j})q^2}
   {(\alpha+\beta+\gamma)+(\alpha+\beta j+\gamma\bar{j})q+(\alpha+\beta\bar{j}+\gamma j)q^2} \\
  k_3 &=&  (\alpha+\beta+\gamma)+(\alpha+\beta j+\gamma\bar{j})q+(\alpha+\beta\bar{j}+\gamma j)q^2.
\end{eqnarray*}
Therefore
 \begin{eqnarray*}
  b_{3n} &=& \Bigl(\frac{(1-k_1q)(k_2-q)}{(\alpha+\beta+\gamma)(\alpha+\beta\bar{j}+\gamma j)}\Bigr)^n
    \frac{(1-q^3)^n}{k_3^n(k_1q;q^3)_n(k_2q^2;q^3)_n(q^3;q^3)_n }\\
  b_{3n+1} &=& \Bigl(\frac{1-k_1q}{(\alpha+\beta+\gamma)}\Bigr)^{n+1}\Bigl(\frac{k_2-q}{\alpha+\beta\bar{j}+\gamma j}\Bigr)^n
    \frac{(1-q^3)^n}{k_3^n(k_1q;q^3)_{n+1}(k_2q^2;q^3)_n(q^3;q^3)_n} \\
  b_{3n+2} &=&\Bigl(\frac{(1-k_1q)(k_2-q)}{(\alpha+\beta+\gamma)(\alpha+\beta\bar{j}+\gamma j)}\Bigr)^{n+1}
    \frac{(1-q^3)^n}{k_3^n(k_1q;q^3)_{n+1}(k_2q^2;q^3)_{n+1}(q^3;q^3)_n}.
\end{eqnarray*}  
  Thus
  \begin{eqnarray}\label{exp B a_1 nz cas 1}  
  B(t)=\sum_{u=0}^2{(k_2-q)^{\frac{u(u-1)}{2}}t^u\over((1-k_2q^2)(\alpha+\beta\bar{j}+\gamma j))^{\frac{u(u-1)}{2}}(\alpha+\beta+\gamma)^{\frac{u(3-u)}{2}}}\nonumber\\
  \times
  \hypergeomq{3}{2}{0,0,0}{k_1q^{1+\frac{3}{2}u(3-u)},k_2q^{2+\frac{3}{2}u(u-1)}}{q^3}{{(1-k_1q)(k_2-q)(1-q^3)\over k_3(\alpha+\beta+\gamma)(\alpha+\beta\bar{j}+\gamma j)}t^3}.
  \end{eqnarray}
The relation between coefficients and roots of (\ref{equation charc}) suggests that $a_2(a_1^2-a_2)=0$.
\\
$\ast$ {If $a_2=0$}, then $a_3=\frac{q^3}{1-q^3}a_1^3$ and $a_4=\frac{q^3}{1-q^3}a_1^4$. \\
Relation (\ref{expression of Q_1,Q_2,Q_3}) is equivalent to the system 
\begin{eqnarray*}
  a_n &=& q^{n-1}\left(qa_n+a_1(1-q)a_{n-1}+a_1^2(q-1)a_{n-2}+a_1^3a_{n-3}\right) \\
  a_n &=& (jq)^{n-1}\left(jqa_n+a_1(1-jq)a_{n-1}+a_1^2(jq-1)a_{n-2}+a_1^3a_{n-3}\right) \\
  a_n &=& (\bar{j}q)^{n-1}\left(\bar{j}qa_n+a_1(1-\bar{j}q)a_{n-1}+a_1^2(\bar{j}q-1)a_{n-2}+a_1^3a_{n-3}\right).
\end{eqnarray*}
Therefore
$$
  a_{3n+2} = 0, \
  a_{3n+1} = a_1a_{3n},\ \textrm{and}\
  a_{3n} = \frac{q^{3n(n+1)\over2}a_1^{3n}}{(q^3;q^3)_n}.
$$
It follows that
\begin{equation}\label{exp A1 a_1 nz cas 1}
A(t)={1+a_1t\over e_{q^3}((a_1qt)^3)}.
\end{equation}
$\ast$ {If $a_2=a_1^2$,} then $a_3=\frac{q^3}{q^3-1}a_1^3$, $a_4=\frac{q^3}{q^3-1}a_1^4$. By aids of (\ref{expression of Q_1,Q_2,Q_3}), we get
\begin{eqnarray*}
  a_n &=& q^{n-2}\left(q^2a_n+a_1q(1-q)a_{n-1}+a_1^2(1-q)a_{n-2}-a_1^3a_{n-3}\right) \\
  a_n &=& (jq)^{n-2}\left(\bar{j}q^2a_n+a_1q(j-\bar{j}q)a_{n-1}+a_1^2(1-jq)a_{n-2}-a_1^3a_{n-3}\right)  \\
  a_n &=& (\bar{j}q)^{n-2}\left(jq^2a_n+a_1q(\bar{j}-jq)a_{n-1}+a_1^2(1-\bar{j}q)a_{n-2}-a_1^3a_{n-3}\right),
\end{eqnarray*}
which yields,
$$
  a_{3n+2} = a_1a_{3n+1},\ 
  a_{3n+1} = a_1a_{3n},\
  \textrm{and}\ a_{3n} = \frac{(-1)^nq^{3n(n+1)\over2}}{(q^3;q^3)_n}(a_1^3)^n,
$$
which yields to
\begin{equation}\label{exp A2 a_1 nz cas 1}
A(t)={1+a_1t+(a_1t)^2\over e_q(-(qa_1t)^3)}.
\end{equation}
The 2-OPS of Brenke type, with $B(t)$ and $ A(t)$ given respectively by (\ref{exp B a_1 nz cas 1}) and (\ref{exp A2 a_1 nz cas 1}), can be viewed as a generalization of the Chihara orthogonal PS discovered in~\cite{chihara1968orthogonal} and investigated later in~\cite{chihara1971orthogonality}. This family appears to be a new 2-OPS.\\
\underline{$\mathbf{(\mathcal{B}{112}): \ \delta_{1,3}=0\,\,\textrm{and}\,\,\delta_{2,3}\delta_{3,3}\neq0}$.}\\
$\ast$ {If $\delta_{1,2}\neq0$}, we have $x_2^3=x_3^3$ and $x_1^3=x_3^2=x_2^2$, which is impossible.\\
$\ast$ {If $\delta_{1,2}=0$}, then $\delta_{1,1}\neq0$ and $x_1^3=x_2=x_3$. This is again impossible.\\
\underline{$\mathbf{(\mathcal{B}{113}):\ \delta_{1,3}=\delta_{2,3}=0\,\,\textrm{and}\,\,\delta_{3,3}\neq0}$.}\\
$\ast$ {If $\delta_{1,2}\delta_{2,2}\neq0$}, $x_1=-x_2$ and $x_3^2=x_1^3=x_2^3$. That is impossible.\\
$\ast$ {If $\delta_{1,2}=\delta_{2,2}=0$}, $x_1=x_2$. This is also impossible.\\
$\ast$ {If $\delta_{1,2}=0$ and $\delta_{2,2}\neq0$} implies that $x_2=x_1^2,\,x_3=x_1^3$ and $x_3^2=x_2^3$.\\
Put $q=x_1$, and write $(x_1,x_2,x_3)=(q,q^2,q^3)$, $q\neq0,\pm1,j,\bar{j} $\\ 
$\triangleright$ \textit{Determination of $B(t)$ and $A(t)$.}
We have
$\Delta_n=\alpha q^n+\beta q^{2n}+\gamma q^{3n}
$,
and
\begin{eqnarray*}
  r_n &=& \alpha\frac{1-q^{n+1}}{1-q}+\beta\frac{1-q^{2(n+1)}}{1-q^2}+\gamma\frac{1-q^{3n+3)}}{1-q^3} \\
   &=& \frac{1-q^{n+1}}{1-q}\Bigl(k_0+(k_0-\alpha)q^{n+1}+(k_0-\alpha-\frac{\beta}{1+q})q^{2n+2}\Bigr), \\   
\end{eqnarray*}
where $k_0=\alpha+\frac{\beta}{1+q}+\frac{\gamma}{1+q+q^2}$\\
\textbf{If $k_0=0$}, we may write
$r_n=-\alpha\frac{1-q^{n+1}}{1-q}q^{n+1}(1-\nu q^{n+1}),\ \textrm{where}\ -\nu=1+\frac{\beta}{\alpha(1+q)}$,
so that, $b_n=(\frac{q-1}{q\alpha})^n{q^{-{n\choose2}}\over(q;q)_n(q\nu;q)_n}$.
Thus
\begin{equation}\label{exp1 B one A (1)} B(t)=
\hypergeomq{3}{1}{0,0,0}{\nu, q}{q}{\frac{1-q}{\alpha q}t}.
\end{equation}
\textbf{If $k_0\neq0$}, since $r_n$ can be written as 
$
r_n=\frac{1-q^{n+1}}{1-q}k_0\left(1-\lambda q^{n+1}\right)\left(1-\mu q^{n+1}\right),
$
where $\lambda,\mu$ satisfy, $\lambda+\mu=\frac{\alpha}{k_0}-1$ and $\lambda\mu=1-\frac{\alpha}{k_0}-\frac{\beta}{k_0(1+q)},$
then $b_n=\frac{(1-q)^n}{k_0^n(q;q)_n(q\lambda;q)_n(q\mu;q)_n}$ and hence
\begin{equation}\label{exp2 B one A (2)}
 B(t)=
 \hypergeomq{3}{2}{0,0,0}{q\lambda,q\mu}{q}{\frac{(1-q)t}{k_0}}.
\end{equation}
From the equality $\delta_{1,2}=\delta_{1,3}=0$, we deduce readily that 
$$a_2=\frac{q}{q+1}a_1^2\quad \textrm{and} \quad a_3=\frac{1-q}{(1-q^3)(1+q)}q^3a_1^3.$$
 Thus (\ref{expression of Q_1,Q_2,Q_3}) may be written as
\begin{align*}
\frac{A(t)}{A(qt)}=&1+\eta t,\quad \eta=(1-q)a_1,&\\
\frac{A(t)}{A(q^2t)}=&1+a_1(1-q^2)t+a_1^2q(1-q)^2t^2= (1+\eta t)(1+\eta qt)&\\
\frac{A(t)}{A(q^3t)}=&1+a_1(1-q^3)t+a_1^2(1-q^3)(q-q^2)t^2+a_1^3q^3(1-q)^3t^3,&\\
=&(1+\eta t)(1+\eta qt)(1+\eta q^2t).&
\end{align*}
Hence
\begin{equation}\label{exp A both B (1)}
A(t)=(-\eta t;q)_\infty=\frac{1}{e_q(\eta t)}.
\end{equation}
\underline{$\mathbf{(\mathcal{B}{114}):\ \delta_{1,3}=\delta_{2,3}=\delta_{3,3}=0.}$}\\
$\ast$ {If $\delta_{1,2}\delta_{2,2}\delta_{3,2}\neq0$}, then $x_1=-x_2=-x_3$.\\
$\ast$ {If $\delta_{1,2}=0$ and $\delta_{2,2}\delta_{3,2}\neq0$}, we get $x_1^2=x_2=x_3$: impossible.\\
$\ast$ {If $\delta_{1,2}=\delta_{2,2}=0$ and $\delta_{3,2}\neq0$}, we obtain $x_1=x_2$.\\
In each of the three previous cases, a solution does not exist. 
\paragraph*{$\bullet(\mathcal{B}{12}):\ \alpha\beta\gamma x_2x_3\neq0$ and $x_1=1$.} 
Writing (\ref{expression of Q_1,Q_2,Q_3}) for $x_2, x_3$ and letting $x_1\to1$, we get for $i=1,2$:\\
\centerline{
$
\frac{A(t)}{A(x_it)}= 1+\delta_{i,1}t+\delta_{i,2}t^2+\delta_{i,3}t^3, \ \
\frac{A'(t)}{A(t)}= a_1+(2a_2-a_1^2)t+(a_1^3-3a_1a_2+3a_3)t^2.
$}  
Such a situation can not occur.
\paragraph*{$\bullet(\mathcal{B}{13}):\ \alpha\beta\gamma x_1x_2x_3=0$.}
We suppose that $\gamma x_3=0$ and $\gamma\neq0$, then 0 is a solution for all characteristic equations associated with the relation (\ref{recurrence relation}). So that $a_m=0,\ m=4,5\ldots$, which is impossible by virtue of Proposition \ref{proposition regularity}.\\ 
We deduce that, $\gamma x_3=0$ occurs if and only if $\gamma=0$ and necessary $\Delta_n$ takes the form
\begin{equation}
\Delta_n=\alpha x_1^n+\beta x_2^n.
\end{equation}
\underline{$\mathbf{(\mathcal{B}{131}):\ \alpha\beta x_1 x_2\neq0,\ x_1\neq1\ and \ x_2\neq1}$.}\\
Since the sequences $(x_1^n)_n$ and $(x_2^n)_n$ are linearly independent and after substituting $\Delta_n$ into (\ref{th}), we deduce that $\{P_n\}_{n\geq 0}$ is a 2-OPS if and only if
\begin{equation}\label{expression of Q_1,Q_2}
\frac{A(t)}{A(x_it)}=1+\delta_{i,1}t+\delta_{i,2}t^2+\delta_{i,3}t^3 =Q_i(t),\quad i=1,2,
\end{equation}
where $\delta_{i,s}\quad i=1,2\,\,s=1,2,3$ are defined in (\ref{expression delta_ij}).\\
 From (\ref{expression of Q_1,Q_2}) we get the necessary following condition 
\begin{equation}\label{relation Q1Q2}
Q_1(x_2t)Q_2(t)=Q_1(t)Q_2(x_1t).
\end{equation}
As well in the case $(\mathcal{B}{11})$, we start a discussion according to the leading coefficients in both sides of (\ref{relation Q1Q2}).\\
$\mathbf{(i)\ (\mathcal{B}{1311}):\ \delta_{1,3}\delta_{2,3}\neq0}$. We have $x_1^3=x_2^3$ and then $(x_1,x_2)=(q,jq),\ q\neq1,j,\bar{j}$\\ 
$\triangleright$ \textit{Determination of $B(t)$ and $A(t)$.} In this case, we have
\begin{equation*}
\Delta_{3n}=(\alpha+\beta)q^{3n},\
\Delta_{3n+1}=(\alpha+j\beta)q^{3n+1},\
\Delta_{3n+2}=(\alpha+\bar{j}\beta)q^{3n+2},
\end{equation*} 
and 
\begin{equation*}
  r_{3n} = \frac{\alpha+\beta}{1-k_1q}\left(1-k_1q^{3n+1}\right), \;
  r_{3n+1}= \frac{\alpha+\beta\bar{j}}{k_2-q}\left(1-k_2q^{3n+2}\right), \;
  r_{3n+2} =  k_3\frac{1-q^{3n+3}}{1-q^3},
\end{equation*}
where 
\begin{eqnarray*}
  k_1 &=& \frac{(\alpha+\beta j)+(\alpha+\beta\bar{j})q+(\alpha+\beta)q^2}{(\alpha+\beta)+(\alpha+\beta j)q+(\alpha+\beta\bar{j})q^2}, \
k_2 = \frac{(\alpha+\beta\bar{j})+(\alpha+\beta)q+(\alpha+\beta j)q^2}{(\alpha+\beta)+(\alpha+\beta j)q+(\alpha+\beta\bar{j})q^2} ,\\
  k_3 &=&  (\alpha+\beta)+(\alpha+\beta j)q+(\alpha+\beta\bar{j})q^2.
\end{eqnarray*}
Therefore
 \begin{eqnarray*}
  b_{3n} &=& \left(\frac{(1-k_1q)(k_2-q)}{(\alpha+\beta)(\alpha+\beta\bar{j})}\right)^n
    \frac{(1-q^3)^n}{k_3^n(k_1q;q^3)_n(k_2q^2;q^3)_n(q^3;q^3)_n}, \\
  b_{3n+1} &=&\Bigl(\frac{1-k_1q}{\alpha+\beta}\Bigr)^{n+1}\Bigl(\frac{k_2-q}{\alpha+\beta\bar{j}}\Bigr)^n
    \frac{(1-q^3)^n}{k_3^n(k_1q;q^3)_{n+1}(k_2q^2;q^3)_n(q^3;q^3)_n}, \\
  b_{3n+2} &=&\left(\frac{(1-k_1q)(k_2-q)}{(\alpha+\beta)(\alpha+\beta\bar{j})}\right)^{n+1}
    \frac{(1-q^3)^n}{k_3^n(k_1q;q^3)_{n+1}(k_2q^2;q^3)_{n+1}(q^3;q^3)_n}.
\end{eqnarray*} 
 Thus
\begin{eqnarray}\label{exp B both A (2)} 
  B(t)=\sum_{u=0}^2{(k_2-q)^{\frac{u(u-1)}{2}}t^u\over((1-k_2q^2)(\alpha+\beta\bar{j}))^{\frac{u(u-1)}{2}}(\alpha+\beta)^{\frac{u(3-u)}{2}}}\nonumber\\
  \times
 \hypergeomq{3}{2}{0,0,0}{k_1q^{1+\frac{3}{2}u(3-u)},k_2q^{2+\frac{3}{2}u(u-1)}}{q^3}{{(1-k_1q)(k_2-q)(1-q^3)\over k_3(\alpha+\beta)(\alpha+\beta\bar{j})}t^3}.
  \end{eqnarray}
To evaluate $A(t)$, compare the coefficients in equation (\ref{expression of Q_1,Q_2}), to get
\begin{equation}\label{system 1}
\left\{
\begin{array}{ccc}
a_n&=&q^{n-3}(q^3a_n+\delta_{1,1}q^2a_{n-1}+\delta_{1,2}qa_{n-2}+\delta_{1,3}a_{n-3}),\ n=3,4,\ldots,\\\\
a_n&=&(jq)^{n-3}(q^3a_n+\delta_{2,1}\bar{j}q^2a_{n-1}+\delta_{2,2}jqa_{n-2}+\delta_{2,3}a_{n-3}),\  n=3,4,\ldots,
\end{array}
\right.
\end{equation}
so that,
\begin{eqnarray}
-(qa_1^2+ja_2)a_{3n+1}+(qa_1^3-(q-j)a_1a_2)a_{3n}&=&-qa_1a_{3n+2},\label{system 2 eq1}\\
-(j+q) q^2a_1a_{3n+1}-(q^2a_2-(j+q)qa_1^2)qa_{3n} &=&-(q^3a_{3n+2}+\mu_1 a_{3n-1}),\label{system 2 eq2}\\
(q^3a_1^2-(\bar{j}+q^2)qa_2)a_{3n-1}+\mu_2 a_{3n-2}&=&q^3a_{3n+1}-q^3a_1a_{3n},\label{system 2 equation 3}
\end{eqnarray} 
where 
\begin{eqnarray*}
\mu_1&=&\frac{1-j}{3}(\delta_{1,3}-\bar{j}\delta_{2,3})=-(j+q)q^2a_1^3+((j+2q)q^2a_1a_2+(1-q^3)a_3,\\
\mu_2&=&\frac{1-\bar{j}}{3}(\delta_{1,3}-j\delta_{2,3})=-q^3a_1^3+(\bar{j}+2q^2)qa_1a_2+(1-q^3)a_3.
\end{eqnarray*}
The determinant associated to the system of equations (\ref{system 2 eq1})-(\ref{system 2 eq2}) is
\begin{equation}\label{determinant}
\mathcal{D}=\begin{vmatrix}
-(qa_1^2+ja_2) & qa_1^3-(q-j)a_1a_2 \\\\
-(j+q) q^2a_1 &-(q^3a_2-(j+q)q^2a_1^2)  
\end{vmatrix}=jq^3a_2(a_2-a_1^2),
\end{equation} 
which advises the following discussion depending on $\mathcal{D}$ is zero or not.\\
$\ast$ {\textbf{$a_2(a_1^2-a_2)\neq0$}}: By virtue of (\ref{system 2 eq1})-(\ref{system 2 eq2}), we get
\begin{equation}\label{solution1 syst 2}\\
a_{3n+1}=\frac{a_1}{a_2}a_{3n+2}+D_1a_{3n-1}\ \textrm{and}\ 
a_{3n}=\frac{1}{a_2}a_{3n+2}-D_2a_{3n-1},
\end{equation}
where $D_1={(q-j)a_2-qa_1^2\over jq^3a_2(a_1^2-a_2)}a_1\mu_1$ and $D_2={ja_2+qa_1^2\over jq^3a_2(a_1^2-a_2)}\mu_1$.\\
Substituting (\ref{solution1 syst 2}) in (\ref{system 2 equation 3}), we obtain 
\begin{equation}\label{equation 3}
(q^3D_1+q^3a_1D_2+q^3a_1^2-(\bar{j}+q^2)qa_2+\mu_2\frac{a_1}{a_2})a_{3n-1}+\mu_2 D_1a_{3n-4}=0 .
\end{equation}
Taking respectively $n=1$ and $n=2$, we get  
\begin{equation}\label{equation cons1}
q^3D_1+q^3a_1D_2+q^3a_1^2-(\bar{j}+q^2)qa_2+\mu_2\frac{a_1}{a_2}=0\ \textrm{and}\
 \mu_2 D_1=0.
\end{equation}
$\star$ {If $\mu_2=0$} then $\delta_{1,3}=j\delta_{2,3}$. Equating the coefficients of $t^5$ in (\ref{relation Q1Q2}), we obtain $(j-1)\delta_{1,3}a_2=0$, which is absurd.\\
$\star$ \textit{If $D_1=0$} then two cases arise:\\
$\circ$ If $\mu_1=0$, this is equivalent to $\delta_{1,3}=\bar{j}\delta_{2,3}$. After equating the coefficients of $t^5$ in (\ref{relation Q1Q2}), we obtain $(1-j)\delta_{1,3}q(a_2-a_1^2)=0$ which is again absurd.
\\
$\circ$ If $(q-j)a_2-qa_1^2=0$, this yields $\delta_{2,2}=\bar{j}\frac{1-jq}{1-q}\delta_{1,2}$, which allows us to distinguish two sub-cases:
\\
- If $\delta_{1,2}=0$, then $\delta_{2,2}=0$.  Compare the coefficients of $t^4$ and $t^3$ in (\ref{relation Q1Q2}), we obtain
 $$
\delta_{2,3}-\delta_{1,3}=0\ \textrm{and}\
(\bar{j}+q)\delta_{2,3}-j(1+q)\delta_{1,3}=0,
$$
which gives $\delta_{1,3}=\delta_{2,3}=0$. That is impossible.
\\
- If $\delta_{1,2}\neq0$ then $\delta_{2,2}\neq0$. Again, equating the coefficients of $t^5$ and $t^4$ in (\ref{relation Q1Q2}), we obtain $\delta_{1,3}=\delta_{2,3}=(\frac{q}{q-j}a_1)^3$. It follows that, 
$\mu_1=\mu_2=\delta_{1,3}$, $D_2=-\frac{jqa_1}{(q-j)^2}$ and $q^3a_1^2-(\bar{j}+q^2)qa_2=\frac{\bar{j}+jq}{j-q}q^2a_1^2$. \\
Replace all these expressions in (\ref{equation cons1}), we obtain, $\frac{2q^2a_1^2(1-j)q^2}{(q-j)^2}=0$, which is not possible.
 \\
$\ast$ {\textbf{$a_2=0$}}.  From relation (\ref{system 1}), it follows at once that 
\begin{eqnarray*}
a_{3n+1}-a_1a_{3n}=\frac{1}{a_1}a_{3n+2},\
a_{3n+1}-a_1a_{3n}=-a_1^2a_{3n-1}+(-a_1^3+\frac{1-q^3}{q^3}a_3)a_{3n-2},\\
(j+q)q^2a_1(a_{3n+1}-a_1a_{3n})=q^3a_{3n+2}+((1-q^3)a_3-(j+q)q^2a_1^3)a_{3n-1},
\end{eqnarray*}
which obviously yields,
\begin{equation*}
a_{3n+2}=0,\
a_{3n+1}=a_1a_{3n},\textrm{and}\
a_{3n}={(-1)^nq^{\frac{3n(n+1)}{2}}a_1^{3n}\over(q^3;q^3)_n}.
\end{equation*}
Hence
\begin{equation}\label{exp1 A one B}
A(t)={1+a_1t\over  e_{q^3}((a_1qt)^3)}.
\end{equation}
$\ast$ {\textbf{$a_2=a_1^2$}.} As before, relation~(\ref{system 1}) implies that
\begin{eqnarray*}
a_1^2((j+q)a_{3n+1}+ja_1a_{3n})&=&qa_1a_{3n+2}\\
q^2a_1((j+q)a_{3n+1}+ja_1a_{3n})&=&q^3a_{3n+2}+(q^3a_1^3+(1-q^3)a_3)a_{3n-1}\\
q^3(a_{3n+1}-a_1a_{3n})&=&\bar{j}qa_1^2a_{3n-1}-((\bar{j}+q^2)qa_1^3+(1-q^3)a_3)a_{3n-2}.
\end{eqnarray*}
It follows that 
\begin{equation*}
ja_1(a_{3n+1}-a_1a_{3n})-q(a_{3n+2}-a_1a_{3n+1})=
q^2(a_{3n+1}-a_1a_{3n})-\bar{j}a_1^2(a_{3n-1}-a_1a_{3n-2})=0,
\end{equation*}
which leads to
\begin{equation*}
  a_{3n+2} = a_1a_{3n+1},\
  a_{3n+1} = a_1a_{3n},\ \textrm{and}\
  a_{3n} = \frac{(-1)^nq^{3n(n+1)\over2}}{(q^3;q^3)_n}(a_1^3)^n.
\end{equation*}
Thus
\begin{equation}\label{exp2 A one B}
A(t)={1+a_1t+(a_1t)^2\over e_{q^3}(-(qa_1t)^3)}.
\end{equation}
The obtained families are similar, up to a parameter, to these derived in case ($\mathcal{B}111$). We acquire again a 2-OPS of Chihara type. 
\\
$\mathbf{(ii)\ (\mathcal{B}{1312}):\ \delta_{1,3}=0\ \textrm{and}\,\,\delta_{1,2}\delta_{2,3}\neq0.}$
The roots $x_1$ and $x_2$ are related by $x_2^2=x_1^3$. Put $x_1=q^2$, then $(x_1,x_2)=(q^2,q^3),\ q\neq\pm1,\pm j,\pm\bar{j}$ \\
$\triangleright$ \textit{Determination of $B(t)$ and $A(t)$.} 
We have 
\begin{eqnarray*}
\Delta_n&=&\alpha q^{2n}+\beta q^{3n},\ \textrm{and}\\
r_n&=&\alpha\frac{1-q^{2n+2}}{1-q^2}+\beta\frac{1-q^{3n+3}}{1-q^3}
=\frac{(q;q)_{n+1}}{(1-q)^{n+1}}(k_0+k_0q^{n+1}+(k_0-\frac{\alpha}{1+q})q^{2n+2}),
\end{eqnarray*}
where $k_0=\frac{\alpha}{1+q}+\frac{\beta}{1+q+q^2}$.\\
\textbf{If $k_0=0$}, then $b_n=(\frac{q^2-1}{q^2\alpha})^n{((-1)^nq^{n\choose2})^{-2}\over(q;q)_n}$. So,
\begin{equation}\label{exp1 B One A}
B(t)=\hypergeomq{3}{0}{0,0,0}{-}{q}{\frac{q^2-1}{\alpha q^2}t}.
\end{equation}
\textbf{If $k_0\neq0$}, then $b_n=\frac{(1-q)^n}{k_0^n(\lambda q;q)_n(\mu q;q)_n(q;q)_n}$. Thus,
\begin{equation}\label{exp2 B One A}
B(t)=\hypergeomq{3}{2}{0,0,0}{\lambda q,\mu q}{q}{\frac{1-q}{k_0}t},
\end{equation}
where $\lambda$ and $\mu$ are the solutions of $1+r+(1-\frac{\alpha}{k_0(1+q)})r^2=0$.\\
Equating the coefficients of $t^4$ in the identity (\ref{relation Q1Q2}) and taking into account that $(1-q^6)a_3=a_1^3(q^6-q^4)+a_1a_2(q^2+q^4-2q^6),$
we obtain the quadratic equation in $a_2$:
\begin{equation*}
(1+q^2)(1+q^3)^2a_2^2-(q+3q^3+q^5+2q^6-q^7+2q^8)a_1^2a_2+q^3(1+q^3-q^4+q^5)a_1^4=0,
\end{equation*} 
which yields, 
$$a_2=\frac{q}{q+1}a_1^2\ \textrm{or}\ a_2=\frac{1-q^2+q^3}{(1-q+q^2)(1+q^3)}q^2a_1^2.$$
$\ast$ {If $a_2=\frac{q}{q+1}a_1^2$}, then $a_3=\frac{q^3}{(1+q)(1+q+q^2)}a_1^3$, $\delta_{1,2}=q(1-q)^2a_1^2$, 
$\delta_{2,2}=q(1-q)(1-q^3)a_1^2$ and $\delta_{1,3}=q^3(1-q)^3a_1^3$.\\ 
Substitute all these expressions in (\ref{expression of Q_1,Q_2}), to obtain
\begin{align*}
\frac{A(t)}{A(q^2t)}=&1+(1-q^2)a_1t+q(1-q)^2a_1^2t^2=(1+\eta t)(1+\eta qt),\ \eta=(1-q)a_1,&\\
\frac{A(t)}{A(q^3t)}=&1+(1-q^3)a_1t+q(1-q)(1-q^3)a_1^2t^2+q^3(1-q)^3a_1^3t^3&\\
=&(1+\eta t)(1+\eta qt)(1+\eta q^2t).&
\end{align*}
Therefore
\begin{equation}\label{expression A 1 (q^2,q^3)}
A(t)=(-\eta t;q)_\infty=\frac{1}{e_q(\eta t)}.
\end{equation}
$\ast$ {If $a_2=\frac{1-q^2+q^3}{(1-q+q^2)(1+q^3)}q^2a_1^2$}, then $a_3=-\frac{1-q^3+q^4}{(1-q+q^2)(1+q^2+q^4)(1+q^3)}q^5a_1^3$,\\
 $\delta_{1,2}=(\frac{q-1}{1-q+q^2})^2q^3a_1^2$, $\delta_{2,2}=q^2(1-q^2)\frac{1-q^3}{1+q^3}q^2a_1^2$ and $\delta_{2,3}=(\frac{1-q^2}{1+q^3})^3q^6a_1^3$.\\
Substituting the previous expressions into (\ref{expression of Q_1,Q_2}), we get
\begin{align*}
\frac{A(t)}{A(q^2t)}=&1+(1-q^2)a_1t+(\frac{q-1}{1-q+q^2})^2q^3a_1^2t^2=(1+\nu t)(1+\nu q^3t)&\\
\frac{A(t)}{A(q^3t)}=&1+(1-q^3)a_1t+q^2(1-q^2)\frac{1-q^3}{1+q^3}q^2a_1^2t^2+(\frac{1-q^2}{1+q^3})^3q^6a_1^3t^3&\\
=&(1+\nu t)(1+\nu q^2t)(1+\nu q^4t),\quad \nu=\frac{1-q^2}{1+q^3}a_1.&
\end{align*}
This gives the explicit expression,
\begin{equation}\label{exp A both B}
A(t)=\frac{(-\nu t;q)_\infty}{1+\nu qt}.
\end{equation}
$\mathbf{(iii)\ (\mathcal{B}{1313}):\ \delta_{1,3}=\delta_{1,2}=0\,\,\textrm{and}\,\,\delta_{2,3}\neq0.}$ Since $x_2=x_1^3$, then $(x_1,x_2)$ takes the form  $(q,q^3),\, q\neq0,\pm 1,j,\bar{j}$.\\
$\triangleright$ \textit{Determination of $B(t)$ and $A(t)$.}
\begin{eqnarray*}
\Delta_n&=&\alpha q^n+\beta q^{3n}\\
r_n&=&\alpha\frac{1-q^{n+1}}{1-q}+\beta\frac{1-q^{3n+3}}{1-q^3}=\frac{1-q^{n+1}}{1-q}(k_0+(k_0-\alpha)q^{n+1}+(k_0-\alpha)q^{2n+2}),
\end{eqnarray*}
where $k_0=\alpha+\frac{\beta}{1+q+q^2}$. \\
\textbf{If $k_0=0$}, then $ b_n={(1-q)^n(-1)^n(q^{n\choose2})^{-1}\over(\alpha q)^n(q;q)_n(-q;q)_n}$. Thus we find that
\begin{equation}\label{exp1 B L both A}
\quad B(t)=
\hypergeomq{3}{1}{0,0,0}{-q}{q}{\frac{1-q}{\alpha q}t}.
\end{equation}
\textbf{If $k_0\neq0$}, then $b_n={(1-q)^n\over k_0^n(q;q)_n(\lambda q;q)_n(\mu q;q)_n}$. It follows readily that 
\begin{equation}\label{exp2 B L both A}
B(t)=
\hypergeomq{3}{2}{0,0,0}{{\lambda q,\mu q}}{q}{\frac{1-q}{k_0}t} ,
\end{equation}
where $\lambda$ and $\mu$ are the solutions of the equation $1+(1-\frac{\alpha}{k_0})r+(1-\frac{\alpha}{k_0})r^2=0$.\\
Now, the condition $\delta_{1,2}=\delta_{1,3}=0$, leads to $a_2=\frac{q}{1+q}a_1^2$ and $a_3=\frac{1-q}{(1+q)(1-q^3)}a_1^3$.\\
From (\ref{expression of Q_1,Q_2}), we have
\begin{align*}
\frac{A(t)}{A(qt)}=&1+\eta t,\ \eta=(1-q)a_1,\ \textrm{and}&\\
\frac{A(t)}{A(q^3t)}=&1+(1-q^3)a_1t+q(1-q)(1-q^3)a_1^2t^2+q^3(1-q)^3a_1^3t^3&\\
=&(1+\eta t)(1+\eta qt)(1+\eta q^2t).&
\end{align*}
We deduce that 
\begin{equation}\label{exp A sing B}
A(t)=(-\eta t;q)_\infty=\frac{1}{e_q(\eta t)}.
\end{equation}
The above two cases correspond to the 2-OPS of Little $q$-Laguerre type defined and studied in~\cite{cheikh2011d}.\\
$\mathbf{(iv)\ (\mathcal{B}{1314}):\ \delta_{1,3}=\delta_{2,3}=0}.$
For $i=1,2,\ \delta_{i,3}$ can be written as \\
\centerline{$\delta_{i,3}=(1-x_i)((a_1^3-2a_1a_2+a_3)x_i^2+(a_3-a_1a_2)x_i+a_3)=0$.
}
Since $x_i\neq1$, we deduce that $x_1,\; x_2$ are the roots of $(a_1^3-2a_1a_2+a_3)m^2+(a_3-a_1a_2)m+a_3=~0$. This contradicts (\ref{regularity condition 2}).
\\
\underline{$\mathbf{(\mathcal{B}{132}):\ \alpha\beta x_2\neq0\ and\ x_1=1.}$}\\ Similar to the case $(\mathcal{B}{12})$, there is no solution here.
\\
\underline{$\mathbf{(\mathcal{B}{133}):\ \alpha\beta x_1 x_2=0.}$}\\
Similar computation as case $(\mathcal{B}{13})$, leads to the expression $\Delta_n=\alpha x_1^n,\ \alpha\neq0$.\\
$\triangleright$ \textit{Determination of $B(t)$ and $A(t)$ when $x=1$.}\\ 
On one hand, we have $
\Delta_n=\alpha,\ r_n=\alpha(n+1),\ \textrm{and}\ b_n=\frac{1}{\alpha^nn!},$
which gives
\begin{equation}\label{exp B aal case}
B(t)=\exp(\frac{t}{\alpha}).
\end{equation}
On the other hand, replace $\Delta_n$ in (\ref{th}), we obtain\\
\centerline{
$
\frac{A'(t)}{A(t)}=a_1+(2a_2-a_1^2)t+(3a_3-3a_1a_2+a_1^3)t^2.
$}
Moreover (\ref{regularity condition 2}) implies that $3a_3-3a_1a_2+a_1^3\neq0$. Thus 
\begin{equation}\label{exp A aal case}
A(t)=\exp(a_1t+(a_2-{a_1^2\over2})t^2+(a_3-a_1a_2+{a_1^3\over3})t^3).
\end{equation}
Here we recover again the 2-OPS of Hermite type introduced in~\cite{douak1996}.\\
$\triangleright$ \textit{Determination of $B(t)$ and $A(t)$ when $x_1\neq1$.}\\
Denote $q:=x_1$. We have
\\
\centerline{
$\Delta_n=\alpha q^n,\ r_n=\alpha\frac{1-q^{n+1}}{1-q},
\ \textrm{and}\ b_n={(1-q)^n\over\alpha^n(q;q)_n}$.
}
Thus 
\begin{equation}\label{exp B al case}
B(t)=e_q\Bigl({1-q\over \alpha}t\Bigr).
\end{equation} 
Now, substituting $\Delta_n$ into (\ref{th}), we get
\begin{eqnarray}\label{equation A (Delta alpha q)}
\frac{A(t)}{A(qt)}&=&1+a_1(1-q)t+(a_1^2(q^2-q)+a_2(1-q^2))t^2+(a_1^3(q^2-q^3)\nonumber\\
&-&a_1a_2(q^2+q-2q^3)+a_3(1-q^3))t^3.
\end{eqnarray}
Further, the regularity condition~(\ref{regularity condition 2}) imposes that the leading coefficient of (\ref{equation A (Delta alpha q)}) is not zero.
Factorize (\ref{equation A (Delta alpha q)}), in the form
\\
\centerline{
$
\frac{A(t)}{A(qt)}=(1-\rho t)(1-\lambda t)(1-\mu t),$
}
we deduce that 
\begin{equation}\label{exp A al case}
A(t)=(\rho t;q)_\infty(\lambda t;q)_\infty(\mu t;q)_\infty=\frac{1}{e_q(-\rho t;q)e_q(-\lambda t;q)e_q(-\mu t;q)}.
\end{equation}
In this case, we obtain the $q$-Appell 2-OPS discovered in~\cite{zaghouani2005some}, which can be viewed as an extension of AlSalam-Carlitz orthogonal PS.
\paragraph{$(\mathcal{B}2):\ \Delta_n=(\alpha+\gamma n)x_1^n+\beta x_2^n$, $x_1\neq x_2$ and $\gamma\neq0$}.
\paragraph*{$(\mathcal{B}{21}):\ x_1=q\neq1$}.
Substitute $\Delta_n$ into (\ref{th}), we obtain
\begin{eqnarray*}
\frac{A'(t)}{A(qt)}&=&a_1+(2a_2-a_1^2q)t+(a_1^3q^2-a_1a_2(q^2+q)+3a_3)t^2:=Q_2(t)\\
\frac{A(t)}{A(qt)}&=&1+(1-q)\Bigl[a_1t+(-qa_1^2+a_2(1+q))t^2\\
&+&(a_1^3q^2-a_1a_2(q+2q^2)+a_3(1+q+q^2)\Bigr]t^3=:Q_1(t).
\end{eqnarray*}
Combining the two above equations, we get $A(t)=\exp(\int\frac{Q_2(t)}{Q_1(t)}dt)$, which cannot be a solution of the above system.
\paragraph*{$(\mathcal{B}{22}):\ x_1=1$ and $\beta\neq0$.}
Replace $\Delta_n$ in (\ref{th}), we obtain
\begin{align*}
\frac{A'(t)}{A(t)}=&a_1+(2a_2-a_1^2)t+(3a_3-3a_1a_2+a_1^3)t^2&\\
\frac{A(t)}{A(x_2t)}=&1+(1-x_2)( a_1t+(-x_2a_1^2+a_2(1+x_2))t^2,&\\
+&\Bigl(a_1^3x_2^2-a_1a_2(x_2+2x_2^2)+a_3(1+x_2+x_2^2)\Bigr)t^3,&
\end{align*}
which is impossible.
\paragraph*{$(\mathcal{B}{23}): \ x_1=1$ and $\beta=0$}.
Substitute $\Delta_n$ in (\ref{th}), we obtain
\begin{align}
\frac{A''(t)}{A(t)}=&2a_2+(6a_3-2a_1a_2)t& \label{eq2}\\
\frac{A'(t)}{A(t)}=&a_1+(2a_2-a_1^2)t+(3a_3-3a_1a_2+a_1^3)t^2&\label{eq1}
\end{align}
The solution of (\ref{eq1}) is
\begin{equation}\label{eq3}
A(t)=\exp(1+a_1t+(2a_2-a_1^2)\frac{t^2}{2}+(3a_3-3a_1a_2+a_1^3)\frac{t^3}{3}).
\end{equation}
Combining (\ref{eq2}) and (\ref{eq3}), we conclude that $3a_3-3a_1a_2+a_1^3=0=2a_2-a_1^2=0$,\\
which means  that 1 is a double root of $(a_3-2a_1a_2+a_1^3)m^2+(a_3-a_1a_2)m+a_3=0$.
Therefore $(a_3-2a_1a_2+a_1^3)\Delta_n+(a_3-a_1a_2)\Delta_{n-1}+a_3\Delta_{n-2}=0$, which is in contradiction with
the regularity condition (\ref{regularity condition 2}).
\paragraph{$(\mathcal{B}3):\ \Delta_n=(\alpha+\beta n+\gamma n^2)x_1^n$ and $\gamma\neq0$}.
\paragraph*{$(\mathcal{B}{31}):\ x_1\neq1$.}
In much the same manner as in case \textbf{$(\mathcal{B}{21})$}, we deduce that there is no solution in this case. 
\paragraph*{$(\mathcal{B}{32}):\ x_1=1$.}\ \\
$\triangleright$ \textit{Determination of $B(t)$ and $A(t)$.}\\
We have 
$\Delta_n=\alpha+\beta n+\gamma n^2,\ b_n=(\frac{3}{\gamma})^n{1\over n!(\lambda)_n(\mu)_n}$, and \\
$r_n=\alpha(n+1)+\beta{n(n+1)\over2}+\gamma{n(n+1)(2n+1)\over6}
$,
where $\lambda\mu=\frac{3\alpha}{\gamma}$ and $\lambda+\mu=-\frac{3\beta+\gamma}{2\gamma}$.\\
So that we obtain
\begin{equation}\label{exp B last case}
 B(t)=
\hypergeom{0}{2}{-}{\lambda,\mu}{\frac{3t}{\gamma}} .
\end{equation}
Substitute $\Delta_n=\alpha+\beta n+\gamma n^2$ into (\ref{th}), we obtain
 \begin{eqnarray*}
 \frac{A'(t)}{A(t)}&=&a_1+(2a_2-a_1^2)t+(3a_3-3a_1a_2+a_1^3)t^2\\
 \frac{A''(t)}{A(t)}&=&2a_2+(6a_3-2a_1a_2)t,\ 
 \frac{A^{(3)}(t)}{A(t)}=6a_3.
\end{eqnarray*} 
We have a solution if and only if $a_2=\frac{1}{2}a_1^2$, $a_3=\frac{1}{6}a_1^3$. 
Finally, we conclude that
\begin{equation}\label{exp A last case}
A(t)=\exp(a_1t).
\end{equation}
The obtained PS agrees with the 2-OPS of the Laguerre type defined and investigated in~\cite{cheikh2000classical}.
\subsubsection*{Concluding remarks:} In this paper, we characterize and classify $d$-OPS of Brenke type and we give all GF of such families. We recover some new and known sequences of polynomials.\\
In a forthcoming investigation, we will benefit from the present results to go more deeply into the properties of the obtained polynomials. We will focus on the "new PS" among them (essentially, the 2-OPS of Chihara type obtained in case $\mathcal{B}1311$). For instance, the third-order recurrence relation, the functional vector of 2-orthogonality, a differential or a difference equation, and some expansion formulas among others). \\
On the other hand, we will be interested in the Geronimus problem described at the beginning of the paper in the framework of $d$-orthogonality.
\providecommand{\bysame}{\leavevmode\hbox to3em{\hrulefill}\thinspace}
\providecommand{\MR}{\relax\ifhmode\unskip\space\fi MR }
\providecommand{\MRhref}[2]{%
  \href{http://www.ams.org/mathscinet-getitem?mr=#1}{#2}
}
\providecommand{\href}[2]{#2}

\end{document}